\documentclass[11pt]{article}
\usepackage[margin=1in]{geometry}
\usepackage{amsmath,amssymb,amsthm,booktabs,microtype}
\usepackage{natbib,authblk}
\usepackage[hidelinks]{hyperref}
\usepackage{orcidlink}
\usepackage{graphicx,placeins}
\newtheorem{proposition}{Proposition}
\newtheorem{corollary}{Corollary}
\newcommand{\Ee}{\mathbb E}
\newcommand{\Pp}{\mathbb P}
\DeclareMathOperator{\Var}{Var}
\DeclareMathOperator{\Cov}{Cov}
\title{Noise-adjusted turnover in estimated networks}
\author{Sultan Amed}
\author{Sayantan Banerjee \orcidlink{0000-0001-5414-4817}}
\affil{OM \& QT Area\\Indian Institute of Management Indore}
\date{}
\begin{document}
\maketitle
\begin{abstract}
	Economic networks are often estimated separately over two periods, and changes in their edge sets are interpreted as structural rewiring. Since both networks are estimated, observed turnover also reflects graph-selection error. We study the two-snapshot Hamming-turnover functional under a homogeneous edge-misclassification model. With known sensitivity and specificity and conditional independence of the estimated edge indicators across periods, latent turnover admits a closed-form unbiased adjustment based only on observed turnover and the two estimated graph sizes. We then examine the effects of sparsity, calibration error and dependence across periods. When the number of true links is proportional to \(p\), a false-positive probability of order \(p^{-1}\) generates expected spurious turnover of order \(p\). An error of the same order in calibrating the false-positive probability can likewise leave order-\(p\) bias after adjustment. We also derive the bias induced by cross-period dependence and give sufficient conditions for consistency relative to network size. Numerical results illustrate the finite-sample implications.
\end{abstract}


\section{Introduction}

Changes in estimated network structure are routinely used to describe shifts in economic dependence. In financial applications, networks are estimated over successive periods and changes in their edge sets are interpreted as evidence of changing connectedness, contagion channels or systemic exposure. A natural summary is edge turnover, the number of links that appear or disappear between two estimated graphs. When the graphs themselves are estimated, however, observed turnover need not correspond to structural change.

A link falsely selected in one period and omitted in another contributes an apparent change even when its latent state is unchanged. Sparsity makes this especially consequential. With \(p\) nodes there are \(\binom p2\) possible undirected edges, while many economic networks contain only \(O(p)\) true links. Small edgewise false-positive probabilities can therefore accumulate over a quadratic number of persistent nonedges. As a result, estimation error that is negligible at the level of a single edge can be first-order for aggregate network turnover.

Measurement error in networks and its effect on network summaries have been studied from several directions. \citet{balachandran2017propagation} analyze the propagation of low-rate measurement error to subgraph counts, while \citet{chang2022estimation} develop moment estimators and inference for subgraph densities from noisy network measurements. More general treatments of unreliable network observations include \citet{newman2018unreliable}. For evolving networks, \citet{macdonald2026inference} develop inference under a dynamic noisy-network model; their comparison framework includes directional edge-transition densities whose sum is the normalized Hamming distance. Differential-network methods instead target changes in precision matrices or their supports directly \citep{zhao2014direct}. In financial applications, sampling variation alone can also generate substantial instability in estimated network structure \citep{bhachech2022instability}.

We study the two-snapshot Hamming-turnover problem when edge-misclassification rates are externally calibrated. Under homogeneous sensitivity and specificity and conditional independence across periods, latent turnover admits a closed-form unbiased adjustment based only on observed turnover and the two estimated graph sizes. The correction is explicit; the substantive contribution is to characterize its accuracy under sparsity, calibration error and cross-period dependence.

The sparse-network consequences are sharp. If each latent graph contains \(O(p)\) edges, there are \(\asymp p^2\) persistent nonedges. Under independence across periods and \(q_p\to0\), their expected contribution to observed turnover is of order \(p^2q_p\), and is \(o(p)\) if and only if \(q_p=o(p^{-1})\). Hence, when network size is proportional to \(p\), a false-positive probability of order \(p^{-1}\) generates apparent turnover of the same order as the network itself.

Calibration must be comparably accurate. If the working false-positive probability differs from its true value by \(\delta\), the leading bias of the adjusted turnover is proportional to \(m\delta\), where \(m=\binom p2\). Thus an error of order \(p^{-1}\) in calibrating the false-positive rate can leave order-\(p\) bias after adjustment. We derive the exact finite-sample bias under misspecified sensitivity and specificity, allow dependence between the two estimated edge indicators, and give sufficient conditions for consistency relative to network size. The distinction from dynamic noisy-network approaches lies in the calibration framework and the sparse-network accuracy requirements. Here, the error rates are treated as calibration quantities, with the emphasis on how accurately they must be known before observed turnover can be interpreted as structural change.

\section{Observed and latent network turnover}

Consider two latent undirected graphs \(G_k=(V,E_k)\), \(k=1,2\), on a common set of \(p\) nodes, with \(m=\binom p2\) possible edges. For edge \(e\), let \(A_{ke}\in\{0,1\}\) be its latent state and \(\widehat A_{ke}\) its estimated state. Define latent turnover \(T=\sum_e|A_{1e}-A_{2e}|\), observed turnover \(\widehat T=\sum_e|\widehat A_{1e}-\widehat A_{2e}|\), and estimated graph size \(\widehat M_k=\sum_e\widehat A_{ke}\).

All expectations and covariances are conditional on \(G_1,G_2\). Let \(s\) denote sensitivity and \(q\) the false-positive probability, so specificity is \(1-q\). We assume \(\Pp(\widehat A_{ke}=1\mid G_1,G_2)=q+(s-q)A_{ke}\), with \(d=s-q\neq0\). Thus the marginal error rates are common across edges and periods. Let \(N_{ab}\) be the number of edges with \((A_{1e},A_{2e})=(a,b)\), \(a,b\in\{0,1\}\), so \(T=N_{01}+N_{10}\).

\begin{proposition}\label{prop:oracle}
	Suppose that, for each edge \(e\), \(\widehat A_{1e}\) and \(\widehat A_{2e}\) are conditionally independent given \(G_1,G_2\). Then
	\[
	\Ee\widehat T
	=
	2q(1-q)N_{00}
	+
	2s(1-s)N_{11}
	+
	\{s(1-q)+(1-s)q\}T.
	\]
	Moreover,
	\begin{equation}\label{eq:correction}
		\widetilde T(s,q)
		=
		\frac{
			\widehat T+(s+q-1)(\widehat M_1+\widehat M_2)-2sqm
		}{
			(s-q)^2
		}
	\end{equation}
	satisfies \(\Ee\{\widetilde T(s,q)\}=T\). No independence across distinct edges is required.
\end{proposition}

\begin{proof}
	For a persistent nonedge, persistent edge and changing edge, the disagreement probabilities are \(2q(1-q)\), \(2s(1-s)\), and \(s(1-q)+(1-s)q\), respectively. Summing over the four latent transition classes gives the first identity.
	
	For the second result, define \(Z_{ke}=(\widehat A_{ke}-q)/d\). Then \(\Ee Z_{ke}=A_{ke}\), while conditional independence across periods gives \(\Ee(Z_{1e}Z_{2e})=A_{1e}A_{2e}\). Since \(A_{1e}+A_{2e}-2A_{1e}A_{2e}=|A_{1e}-A_{2e}|\), the statistic \(\sum_e\{Z_{1e}+Z_{2e}-2Z_{1e}Z_{2e}\}\) is unbiased for \(T\). Expanding it yields \eqref{eq:correction}.
\end{proof}

The first two terms in Proposition~\ref{prop:oracle} are generated by edges whose latent states do not change. Equation~\eqref{eq:correction} removes these contributions in expectation. The adjusted estimator need not lie in \([0,m]\) in finite samples, and truncation generally destroys exact unbiasedness. The denominator also shows that the correction becomes unstable as the discriminatory power \(|s-q|\) approaches zero.

The homogeneous-error assumption is substantive. Conditional on the latent graph pair, the misclassification probability of an edge depends only on its own latent state; averaging heterogeneous edge-specific rates does not in general preserve \eqref{eq:correction}. Across-period independence is natural when the two graphs are constructed from independent samples using separately fixed estimation rules, but need not hold for time-series estimates even when the underlying windows do not overlap. Section~3 studies the effect of misspecified calibration and cross-period dependence.

\section{Calibration, dependence and sparsity}

The rates used in \eqref{eq:correction} may differ from their true values, and edge-selection errors may persist across periods. Write \(L=|E_1|+|E_2|\) and \(\Gamma=\sum_e\Cov(\widehat A_{1e},\widehat A_{2e}\mid G_1,G_2)\). Retain the marginal error model of Section~2 but allow dependence between the two estimated indicators for a given edge. For working rates \(s',q'\), let \(d'=s'-q'\neq0\), \(a=(q-q')/d'\), and \(b=d/d'\).

\begin{proposition}\label{prop:bias}
	The correction evaluated at the working rates satisfies
	\begin{equation}\label{eq:bias}
		\Ee\{\widetilde T(s',q')\}
		=
		b^2T+b(1-b-2a)L+2ma(1-a)-\frac{2\Gamma}{(d')^2}.
	\end{equation}
\end{proposition}

\begin{proof}
	Let \(Z'_{ke}=(\widehat A_{ke}-q')/d'\). Then \(\Ee Z'_{ke}=a+bA_{ke}\), and
	\[
	\Ee(Z'_{1e}Z'_{2e})
	=
	(a+bA_{1e})(a+bA_{2e})
	+
	\frac{\Cov(\widehat A_{1e},\widehat A_{2e}\mid G_1,G_2)}{(d')^2}.
	\]
	Summing \(Z'_{1e}+Z'_{2e}-2Z'_{1e}Z'_{2e}\) over edges and using \(\sum_e(A_{1e}+A_{2e})=L\) and \(\sum_eA_{1e}A_{2e}=N_{11}=(L-T)/2\) gives \eqref{eq:bias}.
\end{proof}

With correctly specified marginal rates, \(a=0\) and \(b=1\), so the bias is \(-2\Gamma/d^2\). Positive cross-period covariance therefore suppresses observed disagreement and induces downward bias. Such covariance is plausible when adjacent network estimates are constructed from overlapping rolling windows, since shared observations can induce the same edge-selection errors in both periods. Equation~\eqref{eq:bias} treats the working rates as fixed. If they are estimated from the network data, their stochastic dependence with the estimated graphs must also be accounted for.

A useful special case arises when sensitivity is correctly specified but \(q'=q+\delta\). Under independence across periods,
\begin{equation}\label{eq:delta}
	\Ee\{\widetilde T(s,q+\delta)\}-T
	=
	\frac{d\delta(2T+L-2m)-T\delta^2}{(d-\delta)^2}.
\end{equation}
This follows from \eqref{eq:bias} with \(a=-\delta/(d-\delta)\) and \(b=d/(d-\delta)\). If \(L=O(p)\), \(|d|\) is bounded away from zero and \(\delta\to0\), then
\[
\Ee\{\widetilde T(s,q+\delta)\}-T
=
-\frac{2m}{d}\delta
+
O(p|\delta|+p^2\delta^2).
\]
Since \(m\asymp p^2\), the bias is \(o(p)\) if and only if \(\delta=o(p^{-1})\). Hence a calibration error of order \(p^{-1}\) can leave order-\(p\) bias, comparable to the size of a sparse network. In the sparse regime considered here, this quadratic amplification arises from false-positive calibration because its error is accumulated over \(O(p^2)\) potential nonedges. When \(q'=q\), misspecification of sensitivity enters only through quantities of order \(T\) and \(L\).

The same counting argument governs unadjusted turnover. Under no structural change, a sparse graph has \(N_{00}=m-O(p)\), so persistent nonedges contribute \(2q_p(1-q_p)N_{00}\) to expected observed turnover. If \(q_p\to0\), this contribution is \(o(p)\) if and only if \(q_p=o(p^{-1})\), since \(N_{00}\asymp p^2\). At the boundary \(q_p\asymp p^{-1}\), it remains of order \(p\). Persistent edges contribute \(2s_p(1-s_p)|E_1|\); when \(|E_1|=O(p)\), this term is \(o(p)\) if \(s_p\to1\), but can remain of order \(p\) when sensitivity is bounded away from one. Vanishing false-positive probabilities alone therefore do not guarantee negligible total distortion.

We next give sufficient conditions under which the corrected turnover is consistent relative to network size.

\begin{corollary}\label{cor:consistency}
	Suppose \(L=O(p)\), \(q_p\to0\), and \(|d_p|=|s_p-q_p|\ge\epsilon>0\). Assume independence across periods and independence of the pairs \((\widehat A_{1e},\widehat A_{2e})\) across distinct edges. If \(\widehat s-s_p=o_{\Pp}(1)\) and \(\widehat q-q_p=o_{\Pp}(p^{-1})\), then
	\[
	\frac{\widetilde T(\widehat s,\widehat q)-T}{p}
	\longrightarrow_{\Pp}0.
	\]
	The estimator may be defined arbitrarily on the event \(\widehat s=\widehat q\).
\end{corollary}

\begin{proof}
	At the true rates, write \(\widetilde T(s_p,q_p)=\sum_e\{Z_{1e}+Z_{2e}-2Z_{1e}Z_{2e}\}\), where \(Z_{ke}=(\widehat A_{ke}-q_p)/d_p\). Since \(|d_p|\ge\epsilon\), these summands are uniformly bounded. For a persistent nonedge, the summand is constant unless at least one estimated indicator equals one, an event with probability at most \(2q_p\), so its variance is \(O(q_p)\). At most \(L=O(p)\) edges are present in one or both latent graphs, and their summands have bounded variance. Independence across edges therefore gives
	\[
	\Var\{\widetilde T(s_p,q_p)\}
	=
	O(p^2q_p+p)
	=
	o(p^2),
	\]
	and Proposition~\ref{prop:oracle} with Chebyshev's inequality yields \(\widetilde T(s_p,q_p)-T=o_{\Pp}(p)\).
	
	For calibration error, let \(S=\widehat M_1+\widehat M_2\), \(H=S-2mq_p\), \(U=\widetilde T(s_p,q_p)\), \(h=\widehat s-s_p\), and \(r=\widehat q-q_p\). Since \(\Ee H=d_pL\) and \(\Var(H)=O(p^2q_p+p)\), we have \(H=O_{\Pp}(p)\); also \(U=O_{\Pp}(p)\). Direct subtraction in \eqref{eq:correction} gives
	\[
	\widetilde T(\widehat s,\widehat q)-U
	=
	\frac{
		(h+r)H-2md_pr-2mhr-\{2d_p(h-r)+(h-r)^2\}U
	}{
		(d_p+h-r)^2
	}.
	\]
	Under the stated calibration rates, the numerator is \(o_{\Pp}(p)\), while the denominator is bounded away from zero with probability tending to one. Hence \(\widetilde T(\widehat s,\widehat q)-U=o_{\Pp}(p)\), completing the proof.
\end{proof}

When the number of true links is proportional to \(p\), Corollary~\ref{cor:consistency} gives consistency relative to network size. The calibration requirement is much stronger for the false-positive probability than for sensitivity. In fact, \(\widehat q\) must be accurate to \(o_{\Pp}(p^{-1})\), whereas \(\widehat s\) need only be consistent. Such calibration requires external validation or an identified model based on replicated measurements; two changing graph estimates alone do not generally provide it \citep{chang2022estimation}.

\section{Numerical illustration}

We first quantify how much apparent network change can arise when the latent graph does not change. Let \(|E|=1.5p\), take sensitivity \(s=0.95\), and vary the false-positive probability \(q\). Table~\ref{tab:nochange} reports the exact expectation from Proposition~\ref{prop:oracle}, separating the contribution of persistent nonedges. Even small values of \(q\) can generate large apparent turnover relative to network size. For example, when \(q=0.01\) and \(p=100\), expected observed turnover is \(109.29\), about \(73\%\) of the \(150\)-edge latent network, despite there being no structural change. At \(p=200\), expected turnover exceeds the entire \(300\)-edge network. Even with \(q=0.001\), the turnover-to-network-size ratio rises from \(0.126\) at \(p=50\) to \(0.425\) at \(p=500\).

\begin{table}[!htbp]
	\centering
	\caption{Expected turnover under no structural change}
	\label{tab:nochange}
	\small
	\setlength{\tabcolsep}{4.5pt}
	\begin{tabular}{rrrrrr}
		\toprule
		\(p\) & \(|E|\) & \(q\) &
		Nonedge &
		\(\Ee(\widehat T)\) &
		\(\Ee(\widehat T)/|E|\) \\
		\midrule
		50  & 75  & 0.001 & 2.30   & 9.42   & 0.126 \\
		&     & 0.005 & 11.44  & 18.57  & 0.248 \\
		&     & 0.010 & 22.77  & 29.90  & 0.399 \\
		&     & 0.020 & 45.08  & 52.21  & 0.696 \\[1mm]
		
		100 & 150 & 0.001 & 9.59   & 23.84  & 0.159 \\
		&     & 0.005 & 47.76  & 62.01  & 0.413 \\
		&     & 0.010 & 95.04  & 109.29 & 0.729 \\
		&     & 0.020 & 188.16 & 202.41 & 1.349 \\[1mm]
		
		200 & 300 & 0.001 & 39.16  & 67.66  & 0.226 \\
		&     & 0.005 & 195.02 & 223.52 & 0.745 \\
		&     & 0.010 & 388.08 & 416.58 & 1.389 \\
		&     & 0.020 & 768.32 & 796.82 & 2.656 \\[1mm]
		
		500 & 750 & 0.001 & 247.75  & 319.00  & 0.425 \\
		&     & 0.005 & 1233.80 & 1305.05 & 1.740 \\
		&     & 0.010 & 2455.20 & 2526.45 & 3.369 \\
		&     & 0.020 & 4860.80 & 4932.05 & 6.576 \\
		\bottomrule
	\end{tabular}
	
	\begin{flushleft}
		\footnotesize
		Notes: The latent graph is unchanged across periods and contains \(1.5p\) edges. Sensitivity is \(s=0.95\). ``Nonedge'' denotes the persistent-nonedge contribution \(2q(1-q)N_{00}\).
	\end{flushleft}
\end{table}

We next introduce genuine structural change and examine the role of calibration. Take \(p=100\) and \(|E_1|=|E_2|=150\). Fifteen links are removed and fifteen previously absent links are added, so the latent turnover is \(T=30\). Estimated edge states are generated independently across edges and periods with \(s=0.95\) and \(q=0.01\), using \(5000\) replications. With the correct calibration \(q'=0.010\), the adjusted turnover has mean \(30.098\) and RMSE \(11.328\), whereas the mean unadjusted turnover is \(135.885\) with RMSE \(106.388\).

Table~\ref{tab:calibration} then varies only the false-positive probability used in the correction. The simulated means closely track the exact expectations in \eqref{eq:delta}; for every \(q'\), their difference is less than \(0.62\) Monte Carlo standard errors of the simulated mean. The economically relevant feature is the scale of the calibration effect. Using \(q'=0.009\) yields mean adjusted turnover \(40.225\), whereas \(q'=0.011\) yields \(19.928\). Thus an absolute calibration error of only \(0.001\) changes estimated turnover by roughly one-third of the true \(T=30\). At \(q'=0.008\) or \(q'=0.012\), RMSE is approximately twice that under correct calibration.

\begin{table}[!htbp]
	\centering
	\small
	\caption{Sensitivity of adjusted turnover to false-positive calibration}
	\label{tab:calibration}
	\setlength{\tabcolsep}{6pt}
	\begin{tabular}{cccccc}
		\toprule
		\(q'\) & Simulated mean & Exact mean & Empirical Bias & RMSE & MCSE \\
		\midrule
		0.008 & 50.309 & 50.212 & 20.309  & 23.223 & 0.159 \\
		0.009 & 40.225 & 40.127 & 10.225  & 15.236 & 0.160 \\
		0.010 & 30.098 & 30.000 & 0.098   & 11.328 & 0.160 \\
		0.011 & 19.928 & 19.829 & -10.072 & 15.183 & 0.161 \\
		0.012 & 9.714  & 9.615  & -20.286 & 23.267 & 0.161 \\
		\bottomrule
	\end{tabular}
	
	\begin{flushleft}
		\footnotesize
		Notes: \(p=100\), \(|E_1|=|E_2|=150\), \(T=30\), \(s=0.95\), \(q=0.01\), and \(5000\) replications. Only the working false-positive probability \(q'\) varies. Exact means are obtained from \eqref{eq:delta}; Empirical bias is the simulated mean minus \(T=30\); MCSE is the Monte Carlo standard error of the simulated mean.
	\end{flushleft}
\end{table}

The final experiment treats \(q\) as an estimated calibration quantity rather than a known constant. In each replication, \(K_0\sim\operatorname{Binomial}(n_0,0.01)\) is generated independently of the two estimated graphs and \(\widehat q=K_0/n_0\), where \(n_0\) denotes the number of independent validation classification trials with known nonedge status. Sensitivity remains fixed at \(s=0.95\). 

Table~\ref{tab:validation} shows that calibration uncertainty primarily affects dispersion rather than the center of the estimator in this design. With \(n_0=1000\), RMSE rises to \(34.278\), more than three times the known-\(q\) value, and \(19.42\%\) of adjusted estimates are negative. The effect falls rapidly as the validation sample grows: at \(n_0=10{,}000\), RMSE is \(15.246\), while at \(n_0=100{,}000\) it is \(11.727\), close to the known-calibration benchmark. Note that, the small empirical biases are not to be interpreted as exact unbiasedness after estimating \(q\). The experiment illustrates the finite-sample price of calibration uncertainty.

\begin{table}[!htbp]
	\centering
	\small
	\caption{Adjusted turnover with an estimated false-positive probability}
	\label{tab:validation}
	\setlength{\tabcolsep}{6pt}
	\begin{tabular}{lccccc}
		\toprule
		\(n_0\) & Mean & Empirical Bias & RMSE & MCSE & Negative (\%) \\
		\midrule
		Known \(q\) & 30.098 & 0.098  & 11.328 & 0.160 & 0.30 \\
		1,000       & 29.765 & -0.235 & 34.278 & 0.485 & 19.42 \\
		10,000      & 30.131 & 0.131  & 15.246 & 0.216 & 2.60 \\
		100,000     & 30.091 & 0.091  & 11.727 & 0.166 & 0.42 \\
		1,000,000   & 30.125 & 0.125  & 11.372 & 0.161 & 0.28 \\
		\bottomrule
	\end{tabular}
	
	\begin{flushleft}
		\footnotesize
		Notes: The graph design and \(5000\) replications are as in Table~\ref{tab:calibration}. For each replication, \(K_0\sim\operatorname{Binomial}(n_0,0.01)\) independently of the graph estimates and \(\widehat q=K_0/n_0\). Negative estimates are retained.
	\end{flushleft}
\end{table}

Figure~\ref{fig:rates} illustrates the rate result behind these calculations. It plots the persistent-nonedge contribution relative to network size for \(q_p=p^{-1/2}\), \(p^{-1}\), and \(p^{-3/2}\). The ratio diverges in the first case, remains bounded away from zero at the boundary \(q_p\asymp p^{-1}\), and vanishes in the third. At \(p=500\), the boundary sequence gives a ratio of \(0.660\), close to its limiting value \(2/3\). The figure therefore isolates the sparse-network mechanism in Section~3, that an edgewise false-positive probability can vanish while aggregate apparent turnover remains first-order relative to the network itself.

\begin{figure}[!htbp]
	\centering
	\includegraphics[width=.72\textwidth]{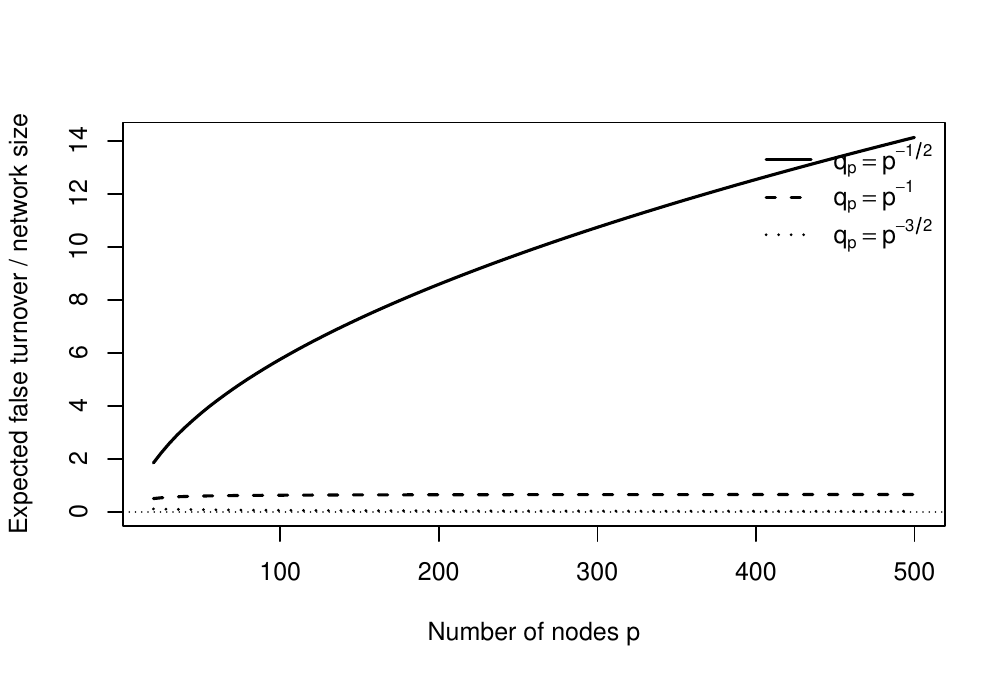}
	\caption{Expected false-positive turnover from persistent nonedges relative to network size. The latent graph contains \(1.5p\) edges and is unchanged across periods.}
	\label{fig:rates}
\end{figure}

\FloatBarrier

\section{Conclusion}

In this paper, we study two-snapshot network turnover under edge misclassification and derive an unbiased adjustment under known error rates and conditional independence across periods. We then show that sparsity imposes much sharper requirements on false-positive control than edgewise accuracy alone would suggest. For turnover error to be negligible relative to an \(O(p)\) network, both false-positive probabilities and their calibration errors must be smaller than the \(p^{-1}\) scale. We also characterize the bias induced by cross-period dependence and give sufficient conditions for consistency relative to network size.

For empirical work, the main implication is that small edgewise error rates do not by themselves guarantee reliable measurement of structural rewiring. The calibration of false positives, together with dependence across network estimates, is first-order for interpreting observed turnover in sparse economic networks.

\bibliographystyle{apalike}
\bibliography{noise_turnover_bib_2}
\end{document}